\documentclass[10pt, onecolumn, letterpaper]{IEEEtran}
\IEEEoverridecommandlockouts
\usepackage{gensymb}
\usepackage{cite}
\usepackage{graphicx}
\usepackage{url}
\usepackage{amsmath}
\usepackage{latexsym,amssymb,epsfig,color,verbatim}
\usepackage{algorithm}
\usepackage[noend]{algpseudocode}
\usepackage{pifont}
\usepackage{slashbox}
\usepackage[tight,footnotesize,center]{subfigure}
\usepackage{color}
\usepackage{cancel}
\usepackage{multirow}
\usepackage{amsthm}

\newcommand{\BEQA}{\begin{eqnarray}}
\newcommand{\EEQA}{\end{eqnarray}}

\newtheorem{theorem}{{\bf Theorem}}

\newtheorem{definition}{{\bf Definition}}

\renewcommand\IEEEkeywordsname{Keywords}
\usepackage{mathtools,amssymb,lipsum}

\usepackage{cuted}
\usepackage{tabu}
\usepackage{textcomp}

\newtheorem{remark}{Remark}

\begin{document}
\bstctlcite{IEEEexample:BSTcontrol}



\title{A Game-Theoretic Framework for Resource Sharing in Clouds }

\author{Faheem Zafari$^{\ast}$, Kin K. Leung$^\ast$, Don Towsley$^\dagger$, Prithwish Basu$^\diamond$, and Ananthram Swami$^\ddagger$\\ ${}^{\ast}$Imperial College London,   ${}^{\dagger}$University of Massachusetts Amherst, ${}^{\diamond}$BBN Technologies,\\ ${}^\ddagger$U.S. Army Research Laboratory\\${}^{\ast}$\{faheem16, kin.leung\}@imperial.ac.uk,  ${}^{\dagger}$towsley@cs.umass.edu, ${}^{\diamond}$prithwish.basu@raytheon.com,\\ ${}^\ddagger$ananthram.swami.civ@mail.mil \thanks{This work was supported by the U.S. Army Research Laboratory and the U.K. Ministry of Defence under Agreement Number W911NF-16-3-0001.  The views and conclusions contained in this document are those of the authors and should not be interpreted as representing the official policies, either expressed or implied, of the U.S. Army Research Laboratory, the U.S. Government, the U.K. Ministry of Defence or the U.K. Government. The U.S. and U.K. Governments are authorized to reproduce and distribute reprints for Government purposes notwithstanding any copy-right notation hereon. Faheem Zafari also acknowledges the financial support by EPSRC Centre for Doctoral Training in High Performance Embedded and Distributed Systems  (HiPEDS, Grant Reference EP/L016796/1), and Department of Electrical and Electronics Engineering, Imperial College London.}}
\maketitle

\begin{abstract}	
Providing resources to different users or applications is fundamental to cloud computing. This is a challenging problem as a cloud service provider may have insufficient resources to satisfy all  user requests. Furthermore, allocating available resources optimally to different applications is also challenging. Resource sharing among different cloud service providers can improve resource availability and resource utilization as certain cloud service providers may have free resources available that can be ``rented'' by other service providers.  However, different cloud service providers can have different  objectives or \emph{utilities}. Therefore, there is a need for a framework that can share and allocate resources in an efficient and effective way, while taking into account the objectives of various service providers that results in a \emph{multi-objective optimization} problem. In this paper, we present a \emph{Cooperative Game Theory} (CGT) based framework for resource sharing and allocation among different service providers with varying objectives  that form a coalition. We show that the resource sharing problem can be modeled as an $N-$player \emph{canonical} cooperative game with \emph{non-transferable utility} (NTU) and prove that the game is convex  for monotonic non-decreasing utilities.  
We propose an  $\mathcal{O}({N})$ algorithm that provides an allocation from the \emph{core}, hence guaranteeing \emph{Pareto optimality}.  We evaluate the performance of our proposed resource sharing framework in a number of simulation settings and show that our proposed framework improves user satisfaction and  utility of  service providers.

\end{abstract}

\section{Introduction}\label{sec:intro}

With the wide-scale proliferation of cloud computing, a large number of individuals and businesses have transitioned towards  the use of cloud computing platforms \cite{shi2017online}. Users (applications)\footnote{Throughout the paper, we use the term ``users'' and ``applications'' interchangeably.} request resources from a cloud service provider for a particular task and  service provider needs to provide  requested resources based on a \emph{Service Level Agreement} (SLA).  The goal for any service provider is to satisfy as many applications as possible by a timely provision  of resources, resulting in an increase in user satisfaction as well as utility of  service providers.  However,  the provider often lacks sufficient available resources and thus allocation of resources to different users is not  straightforward. 
 \par 
Optimal allocation of resources for cloud computing  is a fundamental research problem \cite{xiao2013dynamic,lin2010dynamic,beloglazov2012energy}. Due to the recent advent of  paradigms such as Big Data \cite{john2014big}, Internet of Things (IoT) \cite{atzori2010internet},  and Machine to Machine (M2M) communication, the problem is further exacerbated as the demand for  limited resources has increased further. Even when resources are available, allocating them optimally to different applications requires efficient algorithms. Hence, it is important to ascertain the availability of resources and allocate them optimally to different applications. 

\par To ensure the availability of resources, numerous solutions have been proposed in the literature ranging from buying more resources to meet  peak demands, to creating shared resource pools among geographically co-located service providers. However,  different service providers may have different objectives and  existing resource pool based mechanisms do not consider utilities of different service providers. Ideally,  a resource sharing mechanism should consider objectives of different service providers, which  results in a multi-objective optimization (MOO) problem. Furthermore, while resource sharing or creation of a resource pool may tackle the resource availability problem, it does not address  resource allocation problem, i.e.,  optimally allocating the available resources to different users. Therefore, there is a need for a framework that not only ascertains resource availability by enabling resource sharing among different service providers, but also optimally allocates resources to users, while accounting for the service providers' objectives. 

\par  We present a  cooperative game theory based framework that can be used for resource sharing and allocation in cloud computing. We allow resource sharing among different cloud service providers, while considering the objectives (utilities) of cloud service providers. Rather than considering a specific resource, 
we consider a wide range of resources that are present across different cloud service providers.  
We show that different resource providers can benefit by creating a coalition among  service providers and then allocating resources to users. We also show that the coalition of service providers is stable, where no service  provider will incur any loss (when compared with working alone). 
Also, we obtain a Pareto optimal (details in Section~\ref{sec:prelim}) solution to our multi-objective resource sharing problem.

The main contributions of this paper are:
\begin{enumerate}
	\item We propose a cooperative game theory (CGT) based multi-objective resource sharing and allocation framework for cloud computing that guarantees \emph{Pareto} optimality and consider different  objectives of the service providers that share their resources.
	\item We show that the resource sharing and allocation problem can be modeled as a \emph{canonical} game with non-transferable utility. Furthermore,  when service provider objectives are represented by monotone non-decreasing utilities, we show that our canonical game is convex. Due to the convexity of game, the \emph{core} is non-empty and the grand coalition formed by all service providers  is stable.
\item We address the problem of obtaining an allocation from the core by proposing a computationally efficient  $\mathcal{O}(N)$ algorithm  that is guaranteed to provide an allocation from the core. 
	\item We evaluate the performance of our proposed framework through extensive experiments and show that  resource sharing can improve  utility of all service providers and also increase user satisfaction.  
\end{enumerate}
The paper is further structured as: 
  Section~\ref{sec:prelim} provides a primer on multi-objective optimization, cooperative game theory, canonical games, and the core. Section~\ref{sec:sysmodel} presents the system model used in this paper. It also presents the resource sharing and allocation problem along with theoretical properties of the cooperative game used to obtain the Pareto optimal solution. Section~\ref{sec:exp_results} presents our experimental results. Section \ref{sec:related} provides a review of different solutions proposed in literature for  resource allocation,  while Section~\ref{sec:conclusion} concludes the paper.

\section{Preliminaries}\label{sec:prelim}
In this section, we present a discussion on multi-objective optimization, cooperative game theory, and the \emph{core}. 

\subsection{Multi-Objective Optimization}
 For $m$ inequality constraints and $p$ equality constraints, MOO identifies a vector $\boldsymbol{x}^*=[x_1^*,x_2^*,\cdots ,x_t^*]^T$ that optimizes a vector function
	\begin{equation}\label{eq:vectorobjective}
	\centering 
	\bar{f}(\boldsymbol{x})=[f_1(\boldsymbol{x}),f_2(\boldsymbol{x}),\cdots,f_N(\boldsymbol{x})]^T
	\end{equation}
	such that 
	\begin{align}\label{eq:constraints}
	\centering 
	& g_i(\boldsymbol{x})\geq 0, \; i=1,2,\cdots,m, \\
	&h_i(\boldsymbol{x})=0\;\; i=1,2,\cdots,p. \nonumber
	\end{align}
	where $\boldsymbol{x}=[x_1,x_2,\cdots ,x_t]^T$ is a vector of $t$ decision variables and the feasible set is denoted by $F$.
The fundamental difference between a single objective optimization (SOO) and MOO is that MOO involves a vector of objective functions rather than a single objective function. Therefore, in MOO, the optimal solution is not a single point but a \emph{frontier} of solutions known as \emph{Pareto frontier} or \emph{Pareto boundary} (see \cite{cho2017survey} for details). Some basic definitions related to MOO are: 


\begin{definition}{Pareto Optimality:}
	For any maximization problem,  $\boldsymbol{x}^*$ is $Pareto\; optimal$ if the following holds for every $\boldsymbol{x} \in F$,
	\begin{align}\label{eq:paretoptimality}
	\bar{f}(\boldsymbol{x}^*)\geq\bar{f}(\boldsymbol{x}). 
	\end{align}
	where $\bar{f}(\boldsymbol{x})=[f_1(\boldsymbol{x}),f_2(\boldsymbol{x}),\cdots,f_N(\boldsymbol{x})]^T$ and  
	$\bar{f}(\boldsymbol{x}^*)=[f_1(\boldsymbol{x}^*),f_2(\boldsymbol{x}^*),\cdots,f_N(\boldsymbol{x}^*)]^T$.
\end{definition}

\begin{figure}
\centering
\includegraphics[width=0.4\textwidth]{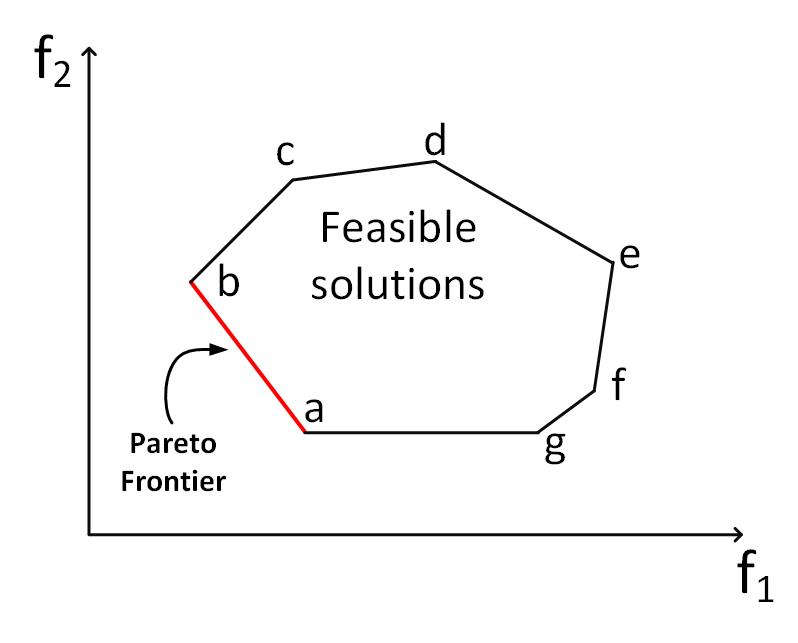}
\caption{MOO with two objective functions}
\protect\label{fig:pareto}
\end{figure}

Figure~\ref{fig:pareto} shows a MOO problem (minimization problem) with two objective functions $f_1$ and $f_2$.  The boundary $\overleftrightarrow{ab}$ is the Pareto frontier  as it consists of all the Pareto optimal solutions. 

\subsection{Cooperative Game Theory}
Cooperative game theory provides a set of analytical tools that assists in understanding the behavior of rational players in a cooperative setting \cite{han2012game}. Players can have agreements among themselves that affect the strategies as well as obtained utilities of game players. Coalition games are one of the basic types of cooperative games that deal with the formation of  coalitions, namely groups of two or more cooperating players. Formally,

\begin{definition}{Coalition Games \cite{han2012game}:}
Any coalition game with \emph{non-transferable utility} (discussed below)  can be represented by the pair $(\mathcal{N},\mathcal{V})$ where $\mathcal{N}$ is the set of players that play the game, while $\mathcal{V}$ is a set of payoff vectors such that \cite{roger1991game}:


\end{definition}
\begin{enumerate}
	\item $\mathcal{V}(S)$ is a closed and convex subset of $\mathbb{R}^S$.
	\item $\mathcal{V}(S)$ is comprehensive, i.e., if we are given payoffs $\mathbf{x} \in \mathcal{V}(S)$ and $\mathbf{y} \in \mathbb{R}^S$ where $\mathbf{y}\leq \mathbf{x}$, then $y \in \mathcal{V}(S)$. In other words, if the members of coalition $S$ can achieve a payoff allocation $\mathbf{x}$, then the players can change their strategies to achieve an allocation $\mathbf{y}$.  
	\item The set $\{\mathbf{x}| \mathbf{x} \in \mathcal{V}(S) \;$ and $\; x_n \geq z_n, \forall n \in S \}$, with $z_n = \max\{y_n|\mathbf{y} \in \mathcal{V}(\{n\}) \}\le \infty \; \forall n \in \mathcal{N}$ is a bounded subset of $\mathbb{R}^S$. In other words, the set of vectors in $\mathcal{V}(S)$ for a coalition $S$ where the coalition members receive a payoff at least as good as working alone (non-cooperatively) is a bounded set. 
\end{enumerate} 

\begin{definition}{Value of a coalition:}
The sum of all players' payoff from a particular pay-off vector for any coalition is known as the value of a coalition. 
\end{definition}
It is worth mentioning that $\mathcal{V}$ is the set of payoff vectors, while ${v}$ is the sum of payoffs that all players get in a particular payoff vector (coalition).
\begin{definition}{Non-Transferable Utility (NTU)\cite{han2012game}:}
	If the total utility of any coalition cannot be  assigned a single real number  or if there is a rigid restriction on  utility distribution among players, then the game has a {non-transferable} utility. 
\end{definition}

\begin{definition}{Characteristic function\cite{han2012game}:}
	The characteristic function for any coalition game with NTU is a function that assigns  a set of payoff vectors, $\mathcal{V}(S) \subseteq \mathbb{R}^S$, where each element of the payoff vector $x_n$ represents a payoff that player $n \in S$ obtains, depending on the selected strategy, within the coalition $S$.
\end{definition}

\begin{definition}{Characteristic form Coalition Games \cite{han2012game}:}
	A coalition game is said to be of {characteristic} form, if the value of coalition $S \subseteq \mathcal{N}$ depends only on the members of coalition. 
\end{definition}

\begin{definition}{Superadditivity of NTU games\cite{roger1991game}:}
A canonical game with NTU is said to be superadditive if the following property is satisfied. 
\begin{align}\label{eq:superadditivity}
v(S_1 \cup S_2) \supset  \{x \in \mathbb{R}^{S_1\cup S_2}|(x_n)_{n \in S_1} \in v(S_1),(x_m)_{m \in S_2}  
\in v(S_2)   \} \;  \forall S_1 \subset \mathcal{N}, S_2 \subset \mathcal{N}, S_1 \cap S_2 =\emptyset. 
\end{align}
\end{definition}
 i.e., if any two disjoint coalitions $S_1$ and $S_2$ form a large coalition $S_1 \cup S_2$, then the coalition $S_1 \cup S_2$ can always give its members the payoff that they would have received in the disjoint coalition  $S_1$ and $S_2$.  
 
\begin{definition}{Canonical Game:}
A coalition game is canonical if it is superadditive and in characteristic form.
\end{definition}
The \emph{core} is a  widely used solution concept for canonical games as discussed below. 

\subsection{Core}
We first define some terms related to the core \cite{han2012game,roger1991game}.

\begin{definition}{Group Rational:}
	A payoff vector $\textbf{x}\in \mathbb{R}^\mathcal{N}$ 
	is group-rational if $\sum_{n \in \mathcal{N}}x_n=v(\mathcal{N})$.
\end{definition}

\begin{definition}{Individually Rational:}
	A payoff vector $\textbf{x}\in \mathbb{R}^\mathcal{N}$ is individually-rational if every player can obtain a payoff no less than acting alone, i.e., $x_n \geq v(\{n\}), \forall n\in \mathcal{N}$.
\end{definition}

\begin{definition}{Imputation:}
	A payoff vector that is both individually and group rational is  an imputation. 
\end{definition}

\begin{definition}{Grand Coalition:}
The coalition formed by all game players in $\mathcal{N}$ is the grand coalition.
\end{definition}

Based on the above definitions, we can now define the core of an NTU canonical coalition game.

\begin{definition}{Core\cite{han2012game}:}
	For any NTU canonical game $(\mathcal{N},\mathcal{V})$, the core is the set of imputations in which no coalition $S\subset\mathcal{N}$ has any incentive to reject the proposed payoff allocation and deviate from the grand coalition to form a coalition $S$ instead. This can be mathematically expressed as 
	\begin{align}\label{eq:coreTU}
	\mathcal{C}_{NTU}&= \{ \mathbf{x} \in \mathcal{V}(\mathcal{N}) | \forall S, \nexists \mathbf{y} \in \mathcal{V}(S), \; such\; that \; y_n > x_n,  	\forall n \in S \}.
	\end{align}
\end{definition}
\begin{remark}\label{rem:paretocore}
	Any payoff allocation from the core is Pareto-optimal as evident from definition of the core. Furthermore, the grand coalition formed is stable, i.e., no two players will have an incentive to leave the grand coalition to form a smaller coalition. 
\end{remark}
However, the core is not always guaranteed to exist. Even if the core exists, it may be very large as it a convex set. Therefore, finding a suitable allocation from the core is challenging.

\section{System Model}\label{sec:sysmodel}

 \begin{figure}
	\centering
	\includegraphics[width=0.48\textwidth]{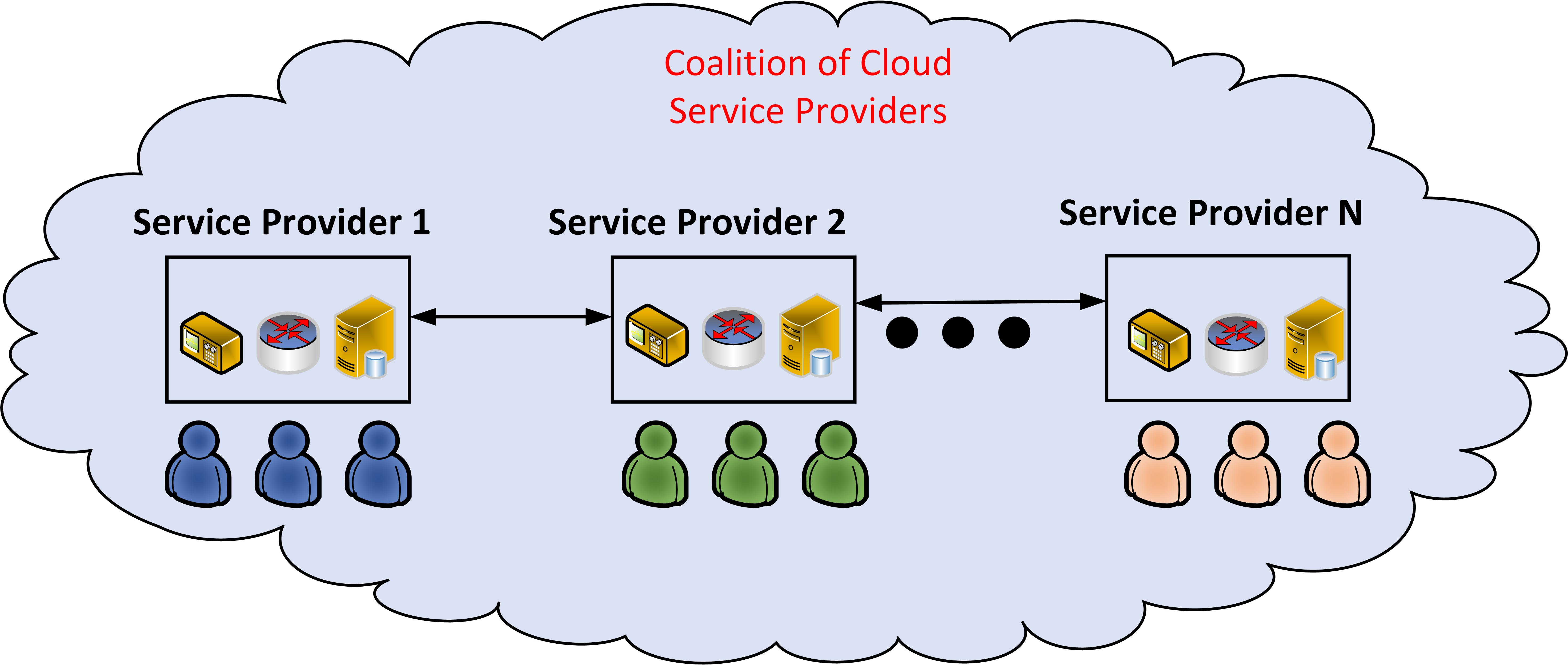}
	\caption{Our system model with multiple cloud service providers}
	\protect\label{fig:system_model}
	\vspace{-0.2in}
\end{figure}

Let $\mathcal{N}=\{1,2, \cdots, N\}$ be the set of all the cloud service providers that act as players in our game. 
We assume that each player has a set of $\mathcal{K}=\{1,2,\cdots, K\}$ different types of resources such as communication, computation and storage resources.  The $n^{th}$ service provider represents its  available resources as $C^n = \{ C_1^{(n)},\cdots, C_K^{(n)}\}$.  $C_k^{(n)}$ is the amount of resources of type $k$ available at service provider $n$. 
The vector $C=\{\sum_{n \in \mathcal{N}}C_{1}^{(n)}, \sum_{n \in \mathcal{N}}C_{2}^{(n)}, \cdots ,\sum_{n \in \mathcal{N}}C_{K}^{(n)}\}$ represents  all available resources at different service providers.  Each service provider $n$ has a set of  native applications $\mathcal{M}_n= \{1,2,\cdots, M_n\}$ that ask for resources. 
The set of all applications that ask for resources from the set of service providers (coalition of service providers) is given by $\mathcal{M}=\mathcal{M}_1\cup\mathcal{M}_2\cdots\cup \mathcal{M}_N , \;$ where $\mathcal{M}_i \cap \mathcal{M}_j=\emptyset, \; \forall i \neq j,$ i.e., each application originally asks only  one service provider for resources. 
Every player  $n \in \mathcal{N}$ has a request (requirement) matrix $R^{(n)}$.
\begin{equation}
\label{eq:R_Req}
R^{(n)}=\Biggl[\begin{smallmatrix}
\mathbf{r^{(n)}_1}\\ 
.\\ 
.\\ 
.\\
\mathbf{r^{(n)}_{{M}_n}}
\end{smallmatrix}\Biggr] = \Biggl[\begin{smallmatrix}
r^{(n)}_{11} & \cdots  &\cdots   & r^{(n)}_{1{K}} \\ 
. & . &.  &. \\ 
. & . &.  &. \\
. & . &.  &. \\
r^{(n)}_{M_{n}1}&\cdots  &\cdots   & r^{(n)}_{{M}_n{K}}
\end{smallmatrix}\Biggr]
\end{equation}
where the $i^{th}$ row corresponds to the $i^{th}$ application while  columns represent different resources, i.e., $r_{ij}^{(n)}$ is the amount of $j^{th}$ resource that application $i \in \mathcal{M}_n$ requests.
Request matrix $R$ is an augmentation of all the request matrices (from all the service providers)  as given below. 
\begin{equation}
\label{eq:R_Reqco}
R^{}=\Biggl[\begin{smallmatrix}
\mathbf{r^{}_1}\\ 
.\\ 
.\\ 
.\\
\mathbf{r^{}_{{M}}}
\end{smallmatrix}\Biggr] = \Biggl[\begin{smallmatrix}
r^{}_{11} & \cdots  &\cdots   & r^{}_{1{K}} \\ 
. & . &.  &. \\ 
. & . &.  &. \\
. & . &.  &. \\
r^{}_{M1}&\cdots  &\cdots   & r^{}_{{M}{K}}
\end{smallmatrix}\Biggr]
\end{equation}
where $r_{ij}$ is the amount of $j^{th}$ resource that application $i \in \mathcal{M}$ requests.

The resource sharing and allocation framework, based on $R$ and  $C$, has to make an allocation decision $\mathcal{X}$ that optimizes the utilities $u_n(\mathcal{X})$ of all the service providers $n \in \mathcal{N}$ and satisfy user requests as well.  The allocation decision $\mathcal{X}$ is  a vector given by $\mathcal{X}=\{X^{(1)},X^{(2)},\cdots, X^{(N)}\}$ that indicates how much of each resource $k \in \mathcal{K}$ is allocated to application $i$ at service provider $n \in \mathcal{N}$. The matrix $X^{(n)}$ is an augmentation of different matrices $\{X_{i\in \mathcal{M}_1}^{(n)}, \cdots,X_{i\in \mathcal{M}_n}^{(n)},\cdots, X_{i\in \mathcal{M}_N}^{(n)} \}$, where $X_{i\in \mathcal{M}_j}^{(n)}$ is the allocation decision  at service provider $n$ for the applications $i$ originally belonging to service provider $j \in \mathcal{N}\backslash n$.  $X_{i\in \mathcal{M}_n}^{(n)}$ is the allocation decision at service provider $n$ for applications $i$ belonging to service provider $n$, i.e., the original applications of the service provider $n$. 
Mathematically, 

\begin{equation}
\label{eq:A_n}
X^{(n)}=\Biggl[\begin{smallmatrix}
\mathbf{x_1^{(n)}}\\ 
.\\ 
.\\ 
.\\
\mathbf{x_{{{M}}}^{(n)}}
\end{smallmatrix}\Biggr] = \Biggl[\begin{smallmatrix}
x_{11}^{(n)} & \cdots  &\cdots   & x_{1{K}}^{(n)} \\ 
. & . &.  &. \\ 
. & . &.  &. \\
. & . &.  &. \\
x_{{{M}}1}^{(n)}&\cdots  &\cdots   & x_{{{M}}{K}}^{(n)}
\end{smallmatrix}\Biggr]
\end{equation}
where $x_{ik}^{(n)}$ is the amount of resource $k\in \mathcal{K}$ belonging to player $n \in \mathcal{N}$ that is allocated  to application $i \in \mathcal{M}$.

\subsection{Optimization Problem}\label{sec:opt_problem}
In this section, we first present the resource allocation problem for a single service provider (no resource sharing with other service providers). Then we present the MOO problem for resource sharing and allocation. 
For a single cloud  $n \in \mathcal{N}$ (not  part of a coalition), the allocation decision matrix $X_{SO}^{(n)}$, where $SO$ stands for single objective,  is given by: 

\begin{equation}
\label{eq:X_n}
X_{SO}^{(n)}=\Biggl[\begin{smallmatrix}
\mathbf{x_1^{(n)}}\\ 
.\\ 
.\\ 
.\\
\mathbf{x_{{{M}_n}}^{(n)}}
\end{smallmatrix}\Biggr] = \Biggl[\begin{smallmatrix}
x_{11}^{(n)} & \cdots  &\cdots   & x_{1{K}}^{(n)} \\ 
. & . &.  &. \\ 
. & . &.  &. \\
. & . &.  &. \\
x_{{{M}_n}1}^{(n)}&\cdots  &\cdots   & x_{{{M}_n}{K}}^{(n)}
\end{smallmatrix}\Biggr].
\end{equation}

The optimization problem is: 
\begin{subequations}\label{eq:opt_single}
	\begin{align}
	\max_{{X}_{SO}^{(n)}}\quad& u_n({X}_{SO}^{(n)})+ \sum_{i \in \mathcal{M}_n}\sum_{k \in \mathcal{K}} \frac{x_{ik}^{(n)}}{r_{ik}^{(n)}} \quad \forall n \in \mathcal{N}, \label{eq:objsingle}\\
	\text{s.t.}\quad & \sum_{i} x_{ik}^{(n)}\leq C_{k}^{(n)} \quad \forall k \in \mathcal{K}, \quad \forall i \in \mathcal{M}_n, \label{eq:singlefirst} \displaybreak[0]\\
	& x_{ik}^{(n)} \leq r^{(n)}_{ik} \quad \forall\; i\in \mathcal{M}_n, k \in \mathcal{K}, \label{eq:singlesecond} \displaybreak[1]\\
	&  x_{ik}^{(n)} \geq 0 \quad \forall\; i\in \mathcal{M}, k \in \mathcal{K}. \label{eq:singlethird} \displaybreak[2]
	\end{align}
\end{subequations}
The goal of this single objective optimization problem for every service provider is   to  maximize its utility by allocating available resources only to its own applications. $\sum_{i \in \mathcal{M}_n}\sum_{k \in \mathcal{K}} {x_{ik}^{(n)}}/{r_{ik}^{(n)}}$ captures the satisfaction of resource requests by applications.  
\par Next, we present the resource allocation problem for different service providers that share their resources with each other. In the resource sharing case,  each service provider aims to  maximize sum of utilities obtained by a) allocating resources to its native applications, i.e., $u_n({X}_{i \in \mathcal{M}_n}^{(n)})$;  b) allocating resources to  applications of other service providers, i.e.,  $u_j^n({X}_{i \in \mathcal{M}_j}^{(n)}),\; \forall j\in \mathcal{N}\backslash n$; and c) satisfying requests of applications. This leads to: 
\begin{subequations}\label{eq:opt_higher}
	\begin{align}
	\max_{\mathcal{X}}\quad&  \bigg(w_nu_n({X}_{i \in \mathcal{M}_n}^{(n)})+\zeta_{n} \sum_{j \in \mathcal{N}\backslash n}  u^{n}_j({X}_{i \in \mathcal{M}_j}^{(n)})+\sum_{i \in \mathcal{M}_n}\sum_{k \in \mathcal{K}}\frac{\sum_{l \in \mathcal{{N}}}x_{ik}^{(l)}}{r_{ik}^{(n)}}\bigg) \quad \forall n \in \mathcal{N}, \label{eq:obj}\\ 
	\text{s.t.}\quad & \sum_{i} x_{ik}^{(n)}\leq C_{k}^{(n)} \quad \forall k \in \mathcal{K},  \forall n \in \mathcal{N}, \forall i \in \mathcal{M}, \label{eq:obj1} \displaybreak[0]\\
	& \sum_{j \in \mathcal{N}}x_{ik}^{(j)} \leq r^{(n)}_{ik} \quad \forall\; i\in \mathcal{M}, k \in \mathcal{K}, n \in \mathcal{N},  \label{eq:obj2}\displaybreak[1]\\
	&  x_{ik}^{(n)} \geq 0 \quad \forall\; i\in \mathcal{M}, k \in \mathcal{K}, n \in \mathcal{N}. \label{eq:obj3} \displaybreak[2]
	\end{align}
\end{subequations}
The first constraint in both optimization problems, \eqref{eq:opt_single}  and \eqref{eq:opt_higher}, indicates that the allocated resources cannot be more than the capacity. The second constraint  in \eqref{eq:singlesecond} and \eqref{eq:obj2} indicates, that the allocated resources should not be more than the required resources.  The last constraint, \eqref{eq:singlethird} and \eqref{eq:obj3},  indicates that allocation decision cannot be negative. 

$u^{n}_j({X}_{i \in \mathcal{M}_j}^{(n)})$ is the utility that domain $n$ receives for sharing its resources with domain $j$. 
$w_n$ is the weight associated with the utility $u_n$ and $\zeta_n$ is the weight affiliated with all $u_j^n$. 
In this paper, we assume that utility of the service providers is a non-decreasing monotone function. 
 We also assume that a subset $\mathcal{N}_1 \subset \mathcal{N}$ of service providers has resource deficit, while a subset $\mathcal{{N}}_2 \subset \mathcal{N}$ has a resource surplus. 

\subsection{Game Theoretic Solution}
As mentioned earlier, 
we model each cloud service provider as a player in our game to obtain the solution for the MOO problem in \eqref{eq:opt_higher}. Let $\mathcal{N}$ be the set of players that can play the resource sharing and allocation game. The \emph{value of coalition} for the game players $S \subseteq \mathcal{N}$ is given by \eqref{eq:payoff_function}, where $\mathcal{F}_S$ is the feasible set. 

\begin{align} \label{eq:payoff_function}
v({S})&=\sum_{{\substack{{n \in S}\\ {\mathcal{X} \in \mathcal{F}_S}}}}    
\bigg(w_n u_n(\mathcal{X})+ \zeta_{n}\sum_{\substack{j \in {S},\\j\neq n}} u^{n}_j(\mathcal{X})+\sum_{i \in \mathcal{M}_n}\sum_{k \in \mathcal{K}}\frac{\sum_{l \in \mathcal{{S}}}x_{ik}^{(l)}}{r_{ik}^{(n)}}\bigg).
\end{align}



\begin{theorem}
Resource allocation and sharing problem (with multiple objectives) for the aforementioned system model  can be modeled as a canonical cooperative game with NTU. 
\end{theorem}
\begin{proof}
To prove this, we need to show that the characteristic function of resource sharing and allocation problem satisfies the following two conditions:
	\begin{itemize}
		\item \textbf{Characteristic form of payoff:} As the utility function in a resource sharing and allocation problem  only relies on the service providers that are part of the coalition, the game or payoff function is of characteristic form.
		\item \textbf{Superadditivity:} 
	For any $S_1, S_2 \subseteq \mathcal{N}$ where $S_1 \cap S_2 =\emptyset$, $S_1,\; S_2 \subset (S_1 \cup S_2)$. The proof follows from definition of monotone utilities.
	
	\end{itemize}
\end{proof}
\begin{definition}[Convex Games]
	A coalition game is said to be convex if and only if for every player $n \in \mathcal{N}$, the marginal contribution of the player is non-decreasing with respect to (W.R.T.) set inclusion. Mathematically, for $S_1\subseteq S_2 \subseteq \mathcal{N}\backslash \{n\}$ 
	\begin{equation}\label{eq:convexgames}
	v(S_1 \cup \{n\})-v(S_1) \leq v(S_2 \cup \{n\})-v(S_2).
	\end{equation} 
\end{definition}
\begin{theorem}\label{thm:convex}
	Our canonical game is convex. 
\end{theorem}
\begin{proof}
	Let us consider two coalition $S_1$ and $S_2$, where $S_1\subseteq S_2 \subseteq \mathcal{N}\backslash\{t\}$, $t \in \mathcal{N}$, and  $\mathcal{F}_{S1}$ and $\mathcal{F}_{S2}$ are the feasible sets for $S_1$ and $S_2$ respectively that it achieves by allocating resources to its own applications. We calculate $v(S_2 \cup \{t\})-v(S_1 \cup \{t\})$ in Equation \eqref{eq:convexityProof}, where $u_t^{(S_1)}$ and $u_t^{(S_2)}$ are the utilities of player $t$ in $S_1$ and $S_2$ respectively.  
	
	\begin{figure*}[!t]
		\normalsize
		\begin{align}\label{eq:convexityProof}
		=& \sum_{\substack{\text{$n \in {S_2}$},\\ \text{$\mathcal{X} \in \mathcal{F}_{S2}$}}}
		\bigg(w_n u_n(\mathcal{X})+ \zeta_{n}\sum_{\substack{j \in {S_2}\cup t,\\ j\neq n}} u^{n}_j(\mathcal{X})+ \sum_{i \in \mathcal{M}_n}\sum_{k \in \mathcal{K}}\frac{\sum_{l \in ({{S}_2\cup t})}x_{ik}^{(l)}}{r_{ik}^{(n)}}\bigg) + w_t u_t ^{(S_2)}(\mathcal{X}) + \zeta_t\sum_{\substack{n \in S_2,\\ n \neq t}} u^{t}_{n}(\mathcal{X})+\nonumber \\ 
	&	\sum_{i \in \mathcal{M}_t}\sum_{k \in \mathcal{K}}\frac{\sum_{l \in {({S}_2}\cup t)}x_{ik}^{(l)}}{r_{ik}^{(t)}} - \Bigg(\sum_{\substack{\text{$n \in {S_1}$},\\ \text{$\mathcal{X} \in \mathcal{F}_{S1}$}}}
		\bigg(w_n u_n(\mathcal{X})+ \zeta_{n}\sum_{\substack{j \in {S_1}\cup t , \\j \neq n}} u^{n}_j(\mathcal{X})+\sum_{i \in \mathcal{M}_n}\sum_{k \in \mathcal{K}}\frac{\sum_{l \in ({{S}_1\cup t})}x_{ik}^{(l)}}{r_{ik}^{(n)}}\bigg) + 
		  w_t u_t ^{(S_1)}(\mathcal{X}) + \nonumber\\ 
		  & \zeta_t\sum_{\substack{n \in S_1,\\ n \neq t}} u^{t}_{n}(\mathcal{X})+ \sum_{i \in \mathcal{M}_t}\sum_{k \in \mathcal{K}}\frac{\sum_{l \in ({{S}_1\cup t})}x_{ik}^{(l)}}{r_{ik}^{(t)}} \Bigg) \nonumber \\
		&= \sum_{\substack{\text{$n \in {S_2}\backslash S_1$},\\ \text{$\mathcal{X} \in \mathcal{F}_{S2}\backslash\mathcal{F}_{S1}$}}}
		\bigg(w_n u_n(\mathcal{X})+ \zeta_{n}\sum_{\substack{j \in {\{(S_2 \cup t) \backslash (S_1\cup t) \}},\\ j \neq n}} u^{n}_j(\mathcal{X})+\sum_{i \in \mathcal{M}_n}\sum_{k \in \mathcal{K}}\frac{\sum_{l \in {\{(S_2 \cup t)\backslash (S_1 \cup t)\}}}x_{ik}^{(l)}}{r_{ik}^{(n)}} \bigg) + w_t u_t ^{(S_2)}(\mathcal{X}) 
		+  \nonumber \\
		&\zeta_t\sum_{\substack{n \in S_2,\\ n \neq t}} u^{t}_{n}(\mathcal{X}) -  \quad  w_t u_t ^{(S_1)}(\mathcal{X}) - \zeta_t\sum_{\substack{n \in S_1, \\ n \neq t}} u^{t}_{n}(\mathcal{X})+\sum_{i \in \mathcal{M}_t}\sum_{k \in \mathcal{K}}\frac{\sum_{l \in \{(S_2 \cup t)\backslash (S_1\cup t)\}}x_{ik}^{(l)}}{r_{ik}^{(t)}}\nonumber \\
		& = v(S_2)- v(S_1)+ w_t u_t ^{(S_2)}(\mathcal{X}) + \zeta_t\sum_{\substack{n \in S_2,\\ n \neq t}} u^{t}_{n}(\mathcal{X}) -  w_t u_t ^{(S_1)}(\mathcal{X}) - \zeta_t\sum_{\substack{n \in S_1,\\ n \neq t}} u^{t}_{n}(\mathcal{X})+\sum_{i \in \mathcal{M}_t}\sum_{k \in \mathcal{K}}\frac{\sum_{l \in \{(S_2 \cup t)\backslash (S_1\cup t)\}}x_{ik}^{(l)}}{r_{ik}^{(t)}}  \nonumber \\
		& \geq v(S_2)- v(S_1)
		\end{align}
		\hrulefill
	\vspace{-0.3in}
	\end{figure*}

	The last inequality in \eqref{eq:convexityProof} follows from the fact that $ w_t u_t ^{(S_2)}(\mathcal{X}) + \zeta_t\sum_{n \in S_2, n \neq t} u^{t}_{l}(\mathcal{X}) -  w_t u_t ^{(S_1)}(\mathcal{X}) - \zeta_t\sum_{n \in S_1, n \neq t} u^{t}_{l}(\mathcal{X})\geq 0$ (monotonicity of the utility).
\end{proof}

\begin{remark}
	The core of any convex game $(\mathcal{N},\mathcal{V})$ is non-empty and large \cite{sharkey1982cooperative}. 
\end{remark}

\begin{remark}\label{lemma:MOOGame}
	Our canonical cooperative game $(\mathcal{N}, \mathcal{V})$ with NTU can be used to obtain the Pareto-optimal solutions for the multi-objective optimization problem given in \eqref{eq:opt_higher}.
\end{remark}

While the existence of core, i.e., core being non-empty, guarantees the grand coalition is stable, finding a suitable allocation from the core is challenging particularly when the core is large (our game). Algorithm \ref{algo:alg1} provides an allocation from the core\footnote{Provided that the aforementioned assumptions hold.}, by solving $|\mathcal{{N}}|+1$ optimization problems. 

\begin{algorithm}
	\begin{algorithmic}[]
		\State \textbf{Input}: $R, C,$ and vector of utility functions of all players $\mathbf{u}$ 
		\State \textbf{Output}: $\mathcal{X}$, ${\mathbf{u}(\mathcal{X})} $
		\State \textbf{Step $1$:}  $ \mathbf{u}(\mathcal{X}) \leftarrow$0,  $ \mathcal{X} \leftarrow$0
		\State \textbf{Step $2$:}
		\For{\texttt{$n \in \mathcal{N}$}}
		\State $ {v(\{n\})}\leftarrow$\texttt{Optimal objective function value in  \eqref{eq:opt_single}} 
		\EndFor
		\State \textbf{Step $3$:}   $ {\mathcal{X}}\leftarrow$\texttt{Optimal allocation decision  from  \eqref{eq:opt_higher2}} \\\quad \quad \quad \quad
		$ {\mathbf{u}(\mathcal{X})}\leftarrow$\texttt{Payoff vector from  \eqref{eq:opt_higher2}} 
	\end{algorithmic}
	\caption{Game-theoretic Pareto optimal allocation }
	\label{algo:alg1}
\end{algorithm}

 \begin{theorem}\label{thm:alg1fromcore}
	The allocation decision obtained using Algorithm \ref{algo:alg1} lies in the core. 
\end{theorem}
\begin{proof}
	To prove the theorem, we need to show that the allocation decision obtained using Algorithm \ref{algo:alg1}:  a) is individually rational; b) is group rational; and c) no players  have the incentive to leave the grand coalition and form another sub-coalition $S \subset \mathcal{N}$. 
	\newline \indent\textbf{Individual Rationality:} For each player $n \in \mathcal{N}$, the solution obtained using \eqref{eq:opt_higher2} is individual rational due to the constraint in \eqref{eq:objh2}. 
	Hence the solution obtained as a result of Algorithm \ref{algo:alg1} is individually rational.
	\newline \indent\textbf{Group Rationality:} The value of the grand coalition $v\{\mathcal{N}\}$ as per Equation \eqref{eq:objh1} is the sum of utilities of all players that they achieve from the pay-off vector  $\mathcal{V}(\mathcal{N})\subseteq \mathbb{R}^{|\mathcal{N}|}$. Hence the allocation obtained from Algorithm \ref{algo:alg1} is group rational.  
	\par Furthermore, due to super-additivity of the game and monotonic non-decreasing nature of the utilities, no subgroup of players have an incentive to form a smaller coalition. Hence Algorithm \ref{algo:alg1} provides a solution from the core.  
\end{proof}
In algorithm \ref{algo:alg1}, we first solve the single objective optimization problem \eqref{eq:opt_single} for all players $n \in \mathcal{N}$. We then solve the problem in \eqref{eq:opt_higher2} that provides the allocation from the core. $v(\{n\})$ in \eqref{eq:objh2} is the payoff a player $n$ receives when working alone.
\begin{subequations}\label{eq:opt_higher2}
	\begin{align}
	\max_{\mathcal{X}}\quad& \sum_{n \in \mathcal{{N}}} \bigg(w_nu_n({X}_{i \in \mathcal{M}_n}^{(n)})+\zeta_{n} \sum_{j \in \mathcal{N}\backslash n}  u^{n}_j({X}_{i \in \mathcal{M}_j}^{(n)})+\sum_{i \in \mathcal{M}_n}\sum_{k \in \mathcal{K}}\frac{\sum_{l \in \mathcal{{N}}}x_{ik}^{(l)}}{r_{ik}^{(n)}}\bigg),\label{eq:objh1}\\ 
	\text{s.t.}\quad & \texttt{constraints in \eqref{eq:obj1}-\eqref{eq:obj3}},  \nonumber \\
& \bigg(w_nu_n({X}_{i \in \mathcal{M}_n}^{(n)})+\zeta_{n} \sum_{j \in \mathcal{N}\backslash n}  u^{n}_j({X}_{i \in \mathcal{M}_j}^{(n)})+
\sum_{i \in \mathcal{M}_n}\sum_{k \in \mathcal{K}}\frac{\sum_{l \in \mathcal{{N}}}x_{ik}^{(l)}}{r_{ik}^{(n)}}\bigg) \geq  {v(\{n\})}, \forall n \in \mathcal{N}\label{eq:objh2}. 
	\end{align}
\end{subequations}

In the next section, we evaluate the performance of our resource sharing and allocation framework. 

\section{Experimental Results}\label{sec:exp_results}

We evaluate the performance of  proposed  resource sharing and allocation framework for  a number of settings as shown in Table \ref{tab:settings}. 
 Each player has three different types of resources ($K=3$), i.e., storage, communication and computation. Without loss of generality, the model can be extended to include other type of resources/parameters. 
We used linear and sigmoidal utilities (see ~\eqref{eq:ut1}) for all the players. However, the results  hold for any monotone non-decreasing utility.
\begin{align}\label{eq:ut1}
u_n(\mathcal{X})=\sum_{i \in \mathcal{M}_n}\bigg(\sum_{k=1}^{{K}}\frac{1}{1+e^{-\mu(x_{ik}^{(n)}-r^{(n)}_{ik})}}\bigg) \quad \forall n \in \mathcal{N}. 
\end{align}
$\mu$ is chosen to be either $0.01$  or $0.1$ to capture the requirements of different applications. 
The request matrices $R^{(n)}, \forall n \in \mathcal{N}$ and the capacity vectors $C^{(n)}, \forall n \in \mathcal{N}$ are randomly generated for each setting within a pre-specified range\footnote{The larger the number of applications in our simulation settings, the larger is range for random number generation.}. To show the advantage of resource sharing, we allocate larger capacities to certain players that, for improving their utilities,  share the available resources with other domains and assist other players in meeting demand. Since all the service providers do not differentiate between their native and applications of other service providers, we set $w_n$ and $\zeta_{n},$ $\forall n \in \mathcal{N}$ to $1$.   
\begin{table}[]
	\centering
	\caption{Simulation network settings.}
	\begin{tabular}{|l|l|}
		\hline
		\textbf{Setting}   & \textbf{Parameters} \\ \hline
		\textbf{1} & $N=3, M_n=3,\forall n \in \mathcal{N}$                    \\ \hline
		\textbf{2} & $N=3, M_n=20,\forall n \in \mathcal{N}$                    \\ \hline
		\textbf{3} & $N=6, M_n=6,\forall n \in \mathcal{N}$                    \\ \hline
		\textbf{4} & $N=6, M_n=20,\forall n \in \mathcal{N}$                   \\ \hline
	\end{tabular}
	\vspace{-0.1in}
	\label{tab:settings}
\end{table}
The simulations were run in \texttt{Matlab R2016b} on a \texttt{Core-i7} processor with \texttt{16 GB RAM}. To solve the optimization problems, we used the \texttt{OPTI-toolbox}. Below, we provide detailed experimental results. 

\subsection{Verification of game-theoretic properties}
In Table \ref{tab:coalition}, we present results  for a 3-player 20-application game that verify different game theoretic properties  such as individually rationality, group rationality,  super additivity and show that the obtained allocation is from the core. The pay-off all players receive in the grand coalition, i.e., $\{1, 2, 3\}$ is at least as good as players $1$, $2$ and $3$ working alone. This shows that the solution obtained using Algorithm \ref{algo:alg1} for the grand coalition is individually rational. Similarly, the value of coalition is the sum of pay-off all players receive is the value of coalition, hence our solution is group rational. Furthermore, with an increase in the coalition size, value of coalition also increases. Hence, the grand coalition has the largest value, which shows the superadditive nature of the game. Also,  no set of players has any incentive to divert from the grand coalition and form a smaller coalition. Hence, the grand coalition is stable and the allocation we obtain using Algorithm \ref{algo:alg1} is from the core.  Similar results were seen for other settings given in Table \ref{tab:settings}. However, we do not include them here due to space constraints.

\begin{table}[]
	\caption{Player payoff in different coalitions for a $3$ player - $20$ application game with $\mu=0.01$}
		\centering
	\begin{tabular}{|l|l|l|l|l|}
		\hline
		\textbf{Coalition}       & \textbf{Player 1} & \textbf{Player 2} & \textbf{Player 3} & \textbf{Value of coalition}  \\ \hline
		\textbf{\{1\}}               & 584.40            & 0.00              & 0.00              & 584.40                   \\ \hline
		\textbf{\{2\}}               & 0             & 90            & 0              & 90                    \\ \hline
		\textbf{\{3\}}               & 0              & 0              & 90             & 90                    \\ \hline
		\textbf{\{1, 2\}}              & 584.40            & 227.15            & 0              & 811.55                   \\ \hline
		\textbf{\{1, 3\}}              & 584.40            & 0              & 229.84            & 814.24                   \\ \hline
		\textbf{\{2, 3\}}              & 0              & 118.91            & 119.50            & 238.41                   \\ \hline
		\textbf{\{1, 2, 3\}} & 584.40            & 205.11            & 202.30            & 991.81             \\ \hline
	\end{tabular}
\label{tab:coalition}
\end{table}

\subsection{Efficacy of the resource sharing framework}
\begin{figure*}
		\centering
	\includegraphics[width=1.05\textwidth]{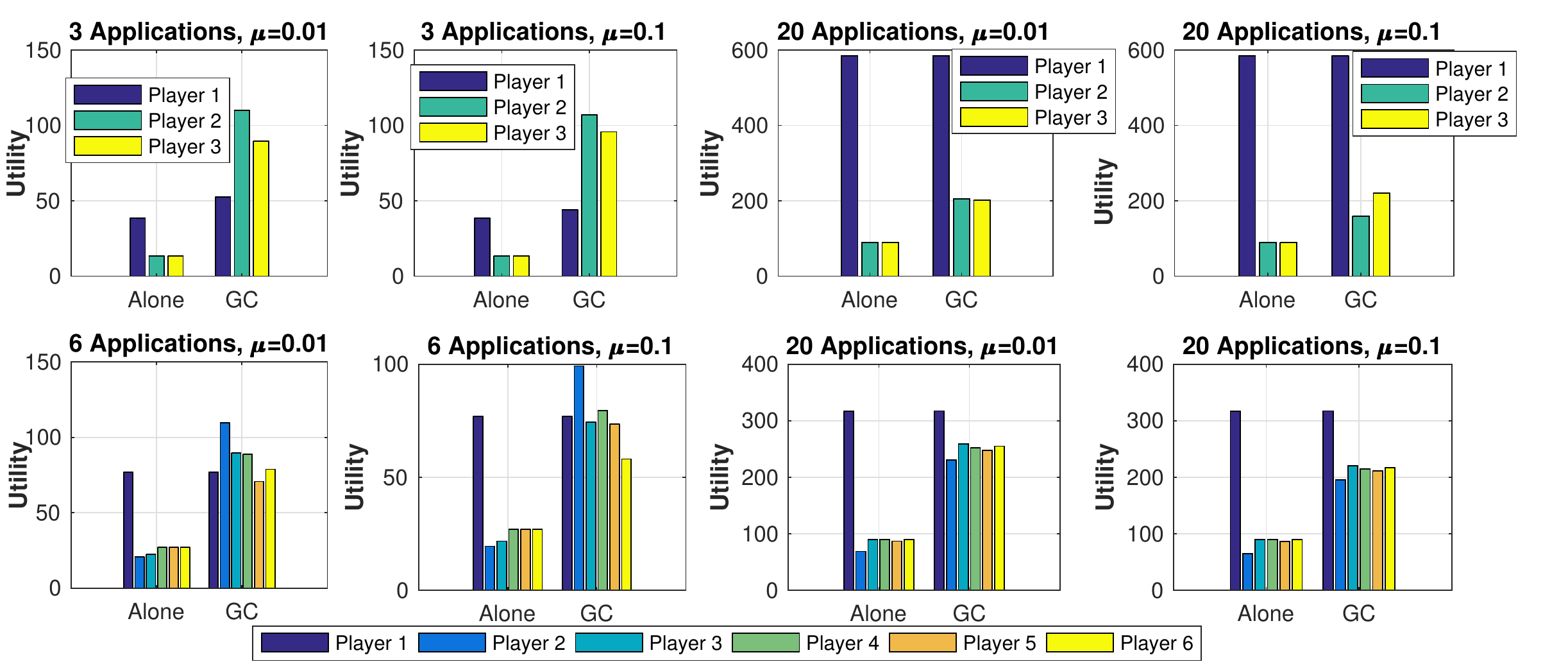}
	\caption{Player utilities working alone and in the grand coalition (GC) for different simulation settings with $\mu=0.01$ and $\mu=0.1$.  }
	\protect\label{fig:util_all}
\end{figure*}
\subsubsection{Player utilities}
Figure \ref{fig:util_all} shows the utility of game players in different simulation settings. It is evident from the figure that the utility of players in grand coalition (GC) in all simulation settings either remains as good as working alone (very small number of cases) or drastically improves with our game theoretic resource sharing framework. This shows that  cloud service providers can greatly benefit by employing a resource sharing scheme. The average improvement in the utility of players seen in our simulations is almost 207\%. However, it is worth mentioning that the improvement greatly depends on parameters such as resource surplus, deficit, and the utilities used.

\subsubsection{User satisfaction}
Satisfying user requests is important for cloud service providers as it is mandated by the SLA. Furthermore, it helps in retaining their users (customers) resulting in long-term profits. Figure \ref{fig:sat01} shows the average user request satisfaction in different player-application settings. For a $3-$ player game with $3$ and $20$ applications respectively, player $1$ has a resource deficit and is not able to satisfy all of its user's request when operating alone. Players $2$ and $3$ have a resource surplus so they can satisfy their users even working alone. When all the three players work in a coalition, the request satisfaction for all three players is 100\%. Similarly, for a $6-$player game with $6$ and $20$ applications respectively, we see that players who could not satisfy their applications alone are able to satisfy all their user requests by joining the coalition and ``renting" resources from the players that have a resource surplus. It is worth mentioning that resource sharing can drastically improve user satisfaction. However, the extent of improvement depends on the nature of the utility function, capacity and resource requests. 
\begin{figure}
		\centering
	\includegraphics[width=0.59\textwidth]{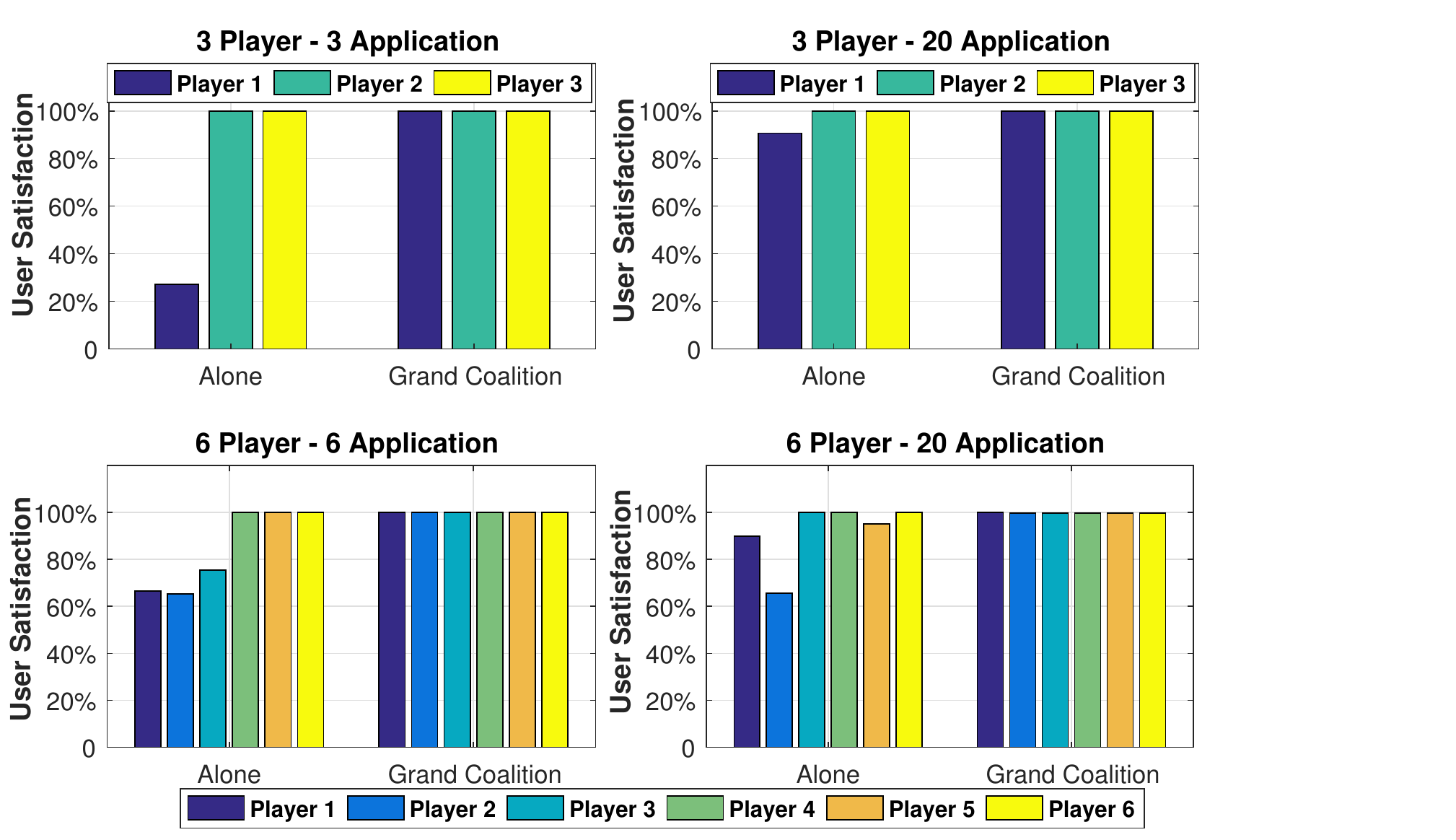}
	\caption{Improvement in user request satisfaction for various player-application settings and $\mu=0.01$.  }
		\protect\label{fig:sat01}
			\vspace{-0.1in}
\end{figure}

\subsubsection{Resource utilization}
Figure \ref{fig:ut01} shows the resource utilization at different players in various simulation settings.  For a $3-$player game with $3$ and $20$ applications respectively, we see that the resource utilization for player $2$ and $3$ (players with resource surplus)  increases in the grand coalition when compared with working alone. This is because these players rent out their resources to player $1$ that may need them resulting in an increased utility. Player $1$ with resource deficit had 100\% resource utilization rate working alone. However, resource utilization for player $1$ in grand coalition either may remain the same or reduce, since its users' requests may be satisfied by other players. Similar results are seen for a $6-$player game with $6$ and $20$ applications  as evident from the figure. 

\begin{figure}
		\centering
	\includegraphics[width=0.59\textwidth]{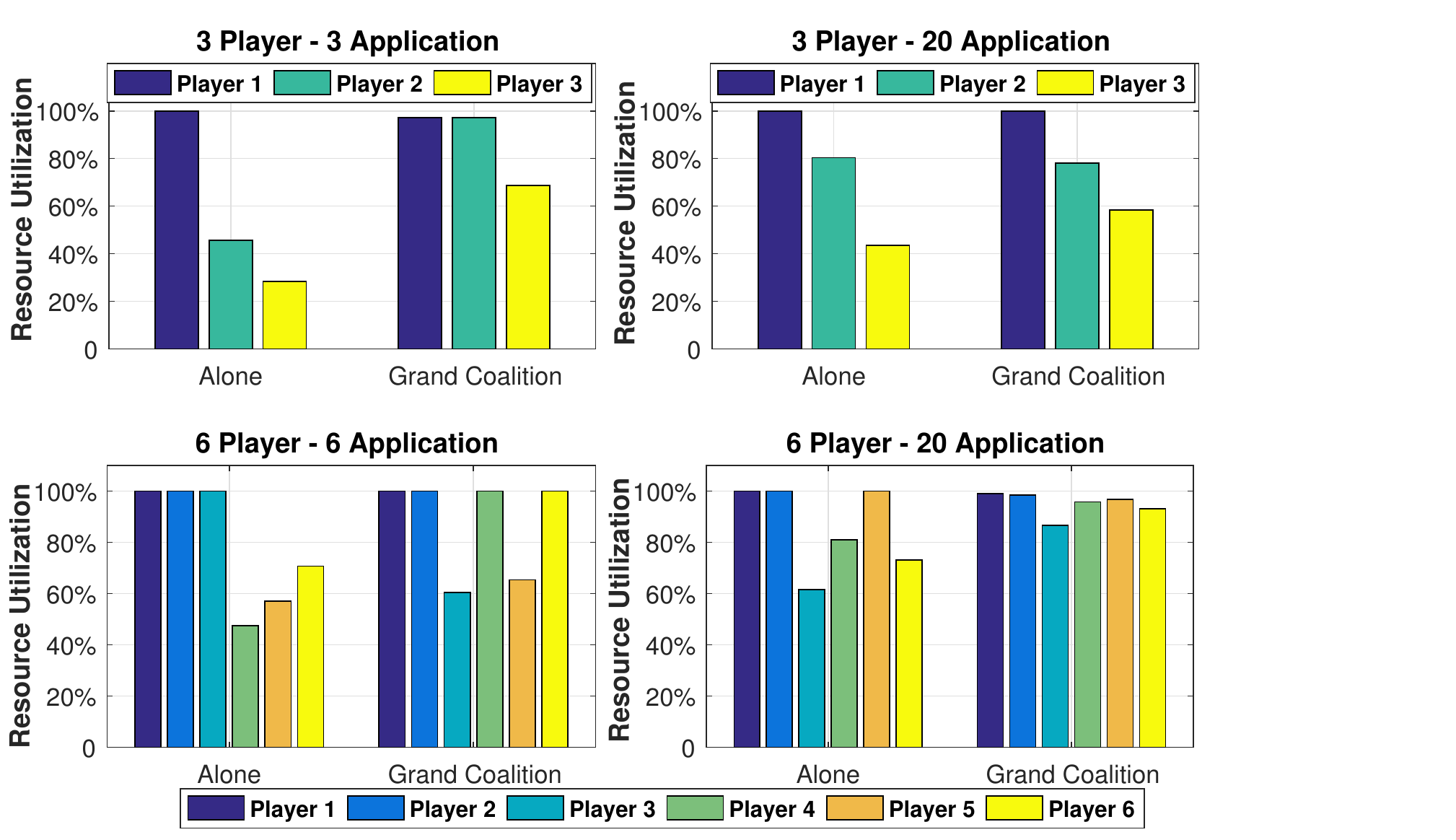}
	\caption{Resource utilization for various player-application settings and $\mu=0.01$.  }
	\protect\label{fig:ut01}
	\vspace{-0.1in}
\end{figure}

\par It is evident from our results that the use of resource sharing and allocation is beneficial for cloud service providers. By allowing different cloud service providers to cooperate, the utility of  service providers increases and resources are utilized in an optimal manner. All the  service providers in the game have the incentive to work together and use their resource capacities in the best possible way. 

\section{Related Work}\label{sec:related}

There have been a number of solutions proposed in literature related to resource availability and resource allocation in cloud computing. Alicherry et al. \cite{alicherry2012network} presented a 2-approximation algorithm for network resource allocation in a distributed cloud environment. 
Ergu et al. \cite{ergu2013analytic} proposed a model for resource allocation in cloud computing environment. 
Tsai et al. \cite{tsai2013optimized} proposed an improved differential evolution algorithm for resource allocation in cloud that combines differential evolution algorithm with Taguchi method. 
Pawar et al. \cite{pawar2013priority} proposed a dynamic resource allocation that considers different SLA parameters and premptable task execution. 
 Wei et al. \cite{wei2010game} proposed a non-cooperative game theory based approach for resource allocation in the cloud. The objective of the algorithm is to maximize fairness among different users. 
Duan et al. \cite{duan2014multi} modeled scheduling in hybrid clouds as a sequential cooperative game. The authors proposed a storage and communication aware algorithm that jointly optimizes the execution time and economic cost of scheduling Bag-of-Tasks work flows. Shi et al. \cite{shi2018shapley} proposed a Shapley value based mechanism for on-demand bandwidth allocation between data centers. They focus on network bandwidth and do not take computing and storage resources into account. Shi et al. \cite{shi2017online} also proposed an online auction mechanism  for dynamically providing virtual clusters in geo-distributed clouds. 
Zafari et al. \cite{zafari2018game} modeled  resource sharing among  mobile edge clouds as a canonical game with transferable utility (TU).
 \par However, our work in this paper differed from \cite{shi2017online,shi2018shapley,wei2010game,alicherry2012network,duan2014multi}. We considered the multi-objective nature of the resource sharing problem and allowed different cloud service providers to share resources and improve their utilities, while satisfying the requests of different users.  
 Our work is closest to \cite{zafari2018game}. However, authors in \cite{zafari2018game} considered the case where each service provider first allocates resources to its own applications and shares the remaining resources with the applications of other service providers. They  only considered the utility of the service provider and did not consider the user satisfaction.   Our framework guarantees Pareto optimality like \cite{zafari2018game}, however, it is more generic than \cite{zafari2018game} as TU games can be modeled as a special case of NTU \cite{roger1991game}. 
 

\section{Conclusions}\label{sec:conclusion}
In this paper, we proposed a cooperative game theory based framework for resource sharing and allocation among cloud service providers. We showed that for a monotonic non-decreasing utility, resource sharing among multiple cloud service providers can be modeled as a canonical game with non-transferable utility. We prove that the game is convex, hence the  core of the game is non-empty.  We proposed an $\mathcal{O}(N)$ algorithm that provides allocation decision from the core. Simulation results showed that our proposed framework improves utility of cloud service providers. Furthermore, request satisfaction of users  also improved.


\bibliographystyle{IEEETran}
\bibliography{refs}  

%

\end{document}